\documentclass[a4paper]{article}

\usepackage[english]{babel}
\usepackage[utf8x]{inputenc}
\usepackage[T1]{fontenc}
\usepackage{graphicx}
\usepackage{caption, subcaption}
\usepackage[title]{appendix}
\usepackage{tikz}
\usepackage{pgfplots}
\usepackage{cutwin}
\graphicspath{ {images/} }
\usetikzlibrary{intersections,decorations.pathreplacing,decorations.markings,calc,angles,quotes,arrows.meta,pgfplots.fillbetween,patterns}

\usepackage[a4paper,top=2.8cm,bottom=3cm,left=2.7cm,right=2.7cm,marginparwidth=1.75cm]{geometry}

\usepackage{amsmath,yhmath,amsfonts}
\usepackage{amsthm}
\usepackage{amssymb}
\usepackage[hidelinks]{hyperref}
\usepackage{cleveref}
\usepackage{epstopdf}
\usepackage{comment}
\usepackage{todonotes}
\usepackage{float}
\usepackage[sort&compress,numbers]{natbib}
\newtheorem{thm}{Theorem}
\newtheorem{lem}[thm]{Lemma}

\theoremstyle{definition}

\newtheorem{rmk}[thm]{Remark}
\numberwithin{equation}{section}
\numberwithin{thm}{section}

\newcommand{\C}{\mathcal{C}}
\newcommand{\Cb}{\mathbb{C}}
\newcommand{\D}{\mathcal{D}}

\newcommand{\R}{\mathbb{R}}

\renewcommand{\S}{\mathcal{S}}

\newcommand{\Z}{\mathbb{Z}}
\newcommand{\p}{\partial}
\renewcommand{\epsilon}{\varepsilon}
\newcommand{\dx}{\: \mathrm{d}}

\newcommand{\uf}{\mathfrak{u}}

\newcommand{\Cf}{\mathfrak{C}}

\newcommand{\frm}{\mathrm{f}}
\newcommand{\trm}{\mathrm{t}}

\newcommand{\nm}{\noalign{\smallskip}}
\newcommand{\ds}{\displaystyle}
\newcommand{\iu}{\mathrm{i}\mkern1mu}

\DeclareMathOperator*{\argmax}{argmax}
\newcommand{\ddp}[2]{\frac{\partial#1}{\partial#2}}

\makeatletter
\newcommand{\neutralize}[1]{\expandafter\let\csname c@#1\endcsname\count@}
\makeatother

\title{Spectral convergence in large finite resonator arrays: the essential spectrum and band structure}

\author{
	Habib Ammari\thanks{\footnotesize Department of Mathematics,
		ETH Zurich, Zurich, Switzerland (habib.ammari@math.ethz.ch).}\and Bryn Davies\thanks{\footnotesize Department of Mathematics, Imperial College London, London, UK (bryn.davies@imperial.ac.uk).} \and Erik Orvehed Hiltunen\thanks{\footnotesize Department of Mathematics, Yale University, New Haven, USA (erik.hiltunen@yale.edu).}}
\date{}

\begin{document}
\maketitle
	
\begin{abstract}
	We show that resonant frequencies of a system of coupled resonators in a truncated periodic lattice converge to the essential spectrum of corresponding infinite lattice. We use the capacitance matrix as a model for fully coupled resonators with long-range interactions in three spatial dimensions. For one-, two- or three-dimensional lattices embedded in three-dimensional space, we show that the (discrete) density of states for the finite system converge in distribution to the (continuous) density of states of the infinite system. We achieve this by proving a weak convergence of the finite capacitance matrix to corresponding (translationally invariant) Toeplitz matrix of the infinite structure. With this characterization at hand, we use the truncated Floquet transform to introduce a notion of spectral band structure for finite materials. This principle is also applicable to structures that are not translationally invariant and have interfaces. We demonstrate this by considering examples of perturbed systems with defect modes, such as an analogue of the well-known interface Su-Schrieffer-Heeger (SSH) model.
\end{abstract}
\vspace{0.5cm}
\noindent{\textbf{Mathematics Subject Classification (MSC2010):} 35J05, 35C20, 35P20.
	
	\vspace{0.2cm}
	
	\vspace{0.5cm}

\section{Introduction}

The spectra of periodic elliptic operators have significant implications for many physical problems and have been studied extensively, as a result. In most cases, this analysis is relatively straightforward, since the spectrum can be decomposed into a sequence of continuous bands using Floquet-Bloch theory \cite{kuchment1993floquet, brillouin1946book}. Meanwhile, the spectra of elliptic operators on finite domains are quite different in nature but are similarly convenient to handle, in most cases. Such problems typically have a discrete spectrum that can be described using a variety of standard techniques. A more subtle question, however, is how to relate these two spectra. 

In the physical and experimental literature on waves in periodic structures, the link between the spectra of finite and infinite structures is made routinely. For example, Floquet-Bloch analysis of infinite structures is commonly used to predict the behaviour of the equivalent truncated version, which can be realised in experiments or simulations. Similarly, measurements from experiments or simulations are often used to recreate the Floquet-Bloch spectra band structure by taking Floquet transforms of spatially distributed data. This work will make the link between these two systems precise, by clarifying how the spectrum of the finite structure converges to that of the corresponding infinite structure, as its size becomes arbitrarily large.

In this work, we will study the capacitance matrix as a model for a system of coupled resonators. This is a model that describes the resonant modes of a system of $N$ resonators in terms of the eigenstates of an $N\times N$ matrix. This matrix is defined in terms of Green's function operators, posed on the boundaries of the resonators. This use of boundary integral formulations allows the model to describe a broad class of resonator shapes \cite{colton2013integral, ammari2009layer}. A crucial feature of the model is that it takes into account long-range interactions between the resonators; interactions between a pair of resonators scale in inverse proportion to the distance between them. This is in contrast to many tight-binding Hamiltonian formulations, which often use nearest-neighbour approximations. The capacitance matrix model was first introduced by Maxwell to model the relationship between the distributions of potential and charge in a system of conductors \cite{maxwell1873treatise}. More recently, it has been shown that the capacitance matrix model also captures the subwavelength resonant modes of a system of high-contrast resonators \cite{ammari2021functional}. Our motivation for using it as the basis for this work is that it serves as a canonical model for a fully coupled system of resonators, which has long-range interactions between the resonators decaying in proportion to the distance. 

An important subtlety of the model considered in this work is that there are no energy sources or damping mechanisms, such that the system is time-reversal symmetric. This means that the capacitance matrix model used in this work is a Hermitian matrix. While generalized capacitance matrix models have been developed for non-Hermitian systems \cite{ammari2020exceptional}, for non-Hermitian models we generally expect drastically different behaviour, whereby the finite and infinite systems have fundamentally different spectra. A symptom of this is the non-Hermitian \emph{skin effect}, whereby all eigenmodes are localised at one end of a finite-sized structure, in certain systems \cite{zhang2022review}.

As the system of finitely many resonators becomes large, the corresponding capacitance matrix also grows in size. Thus, the problem at hand is to understand the asymptotic distribution of the eigenvalues of the capacitance matrix, in the limit that its size becomes arbitrarily large. While this is an open question for the capacitance matrix, similar results exist for other classes of matrices. For example, the asymptotic distribution of the eigenvalues of banded matrices has been studied \cite{bourget2015density, geronimo1988asymptotic}. The capacitance matrix is not banded, but it is ``almost'' banded in the sense that the entries decay in successive off-diagonals. Similarly, there is an established theory describing the limiting spectra of Toeplitz matrices \cite{gray1972asymptotic, tilli1997asymptotic, tyrtyshnikov1998spectra}. Once again, the capacitance matrix is not Toeplitz, but for a periodic resonator array it is known to converge to a doubly infinite matrix that has constant (block) diagonals \cite{anderson}. The crux of this work is to use the properties of the capacitance matrix, as summarised in \emph{e.g.} \cite{diaz2011positivity, ammari2021functional}, to develop an analogous asymptotic eigenvalue distribution theory.

In many of the existing theories on asymptotic eigenvalue distributions, the crucial quantity is the \emph{density of states} (DOS) \cite{bourget2015density, geronimo1988asymptotic}. This is a measure that describes the distribution of eigenmodes in frequency space and is given, for a finite array, by
\begin{equation}
D_\frm(\omega) =\frac{1}{M} \sum_{j=1}^M \delta\Big(\omega-\omega_j^{(M)}\Big),
\end{equation}
where $\omega_1^{(M)}, \omega_2^{(M)}, \dots, \omega_M^{(M)}$ are the $M$ eigenvalues of the $M$-resonator system. One of the main results of this work is showing that $D_\frm$ converges (in the sense of distributions) to the density of states $D$ for the corresponding infinite system (which can be obtained via Floquet-Bloch analysis):
\begin{equation}
D_\frm\to D.
\end{equation}
The proof of this result is based on comparing the capacitance matrix for a finite system of resonators with the finite-sized matrix obtained by taking the matrix arising from the infinite array of resonators and truncating it to be $M$-by-$M$. The key insights are that (i) these two finite-sized matrices converge to the same limit (under a normalised Frobenius norm) as the size becomes large and (ii) the truncated infinite matrix is a block Toeplitz matrix, meaning we can use existing theory. This approach proves to be immensely useful and will allow us to prove not only the convergence of the density of states, but also a theorem demonstrating that this spectral convergence is in fact pointwise. That is, given an eigenvalue $\omega$ of the infinite structure and a positive number $\epsilon$, any sufficiently large finite structure will have an eigenvalue $\omega_\frm$ such that
\begin{equation}
|\omega_\frm-\omega|<\epsilon.
\end{equation}
Further, we will see that a similar convergence result holds for the corresponding eigenvectors, also.

This paper will begin by introducing the capacitance matrix model in \Cref{sec:capacitance}. This is accompanied by an asymptotic derivation of the model from a three-dimensional differential problem (with high-contrast subwavelength resonators) in \Cref{app:cont}, for context. In \Cref{sec:dos}, we will develop the theory needed to prove the convergence of the density of states. This is followed by results on pointwise convergence in \Cref{sec:bands}, which are accompanied by a demonstration of how the Floquet-Bloch spectral bands can be reconstructed from a finite structure using the Floquet transform. Finally, in \Cref{sec:nonperiodic} we present some open questions and possible avenues for future work.

\section{Capacitance matrix model} \label{sec:capacitance}
Throughout this work, we will use a capacitance matrix model for a system coupled resonators. We can view this as a canonical model for coupled resonators with long-range interactions (the interactions are inversely proportional to the distance between the resonators). This model can be derived from first principles to describe either a system of conductors \cite{maxwell1873treatise} or a system of high-contrast resonators, which is summarised in \Cref{app:cont}.

We consider a periodically repeating system in $\R^3$. We take a lattice $\Lambda$ of dimension $d$, where $0<d\leq3$, generated by the lattice vectors $l_1, \dots, l_{d}\in \R^3$. For simplicity, we take $\Lambda$ to be aligned with the first $d$ coordinate axes. We will refer to the three possible lattice dimensions as, respectively, a \emph{chain} of resonators ($d=1$), a \emph{screen} of resonators ($d=2$), and a \emph{crystal} of resonators ($d=3$). We take $Y\subset \R^3$ to be a single unit cell,
$$Y = \begin{cases}
	\{c_1 l_1 + x_2 e_2 + x_3e_3 \mid 0\leq c_1\leq 1, x_{2},x_3 \in \R \}, & d=1, 
	\\
	\{c_1 l_1 + c_2 l_2 + x_3e_3 \mid 0\leq c_1,c_2\leq 1, x_3 \in \R \}, & d=2,
	\\
	\{c_1 l_1 + c_2 l_2 + c_3l_3 \mid 0\leq c_1,c_2,c_3\leq 1 \}, &d=3.
\end{cases}$$
To define the finite lattice, we let $I_r \subset \Lambda $ be all lattice points within distance $r$ from the origin;
$$I_r = \{m \in \Lambda \mid |m| < r \}.$$

The resonators are given by inclusions of a heterogeneous material surrounded by some heterogeneous background medium. We let $D\subset Y$ be a collection of $N$ resonators contained in $Y$,
$$D = \bigcup_{i=1}^N D_i,$$
where $D_n$ are disjoint domains in $Y$ with boundary $\p D_i \in C^{1,s}$ for $s>0$. We will use $D$ to denote the collection of resonators contained within a single unit cell of the periodic lattice. We can subsequently define the \emph{periodic} system $\D$ and the \emph{finite} system $\D_\frm$ of resonators, respectively, as
$$\D=\bigcup_{m\in\Lambda} D+m, \quad \text{and} \quad \D_\frm(r) = \bigcup_{m\in I_r} D + m.$$
Here, $\D$ is the full lattice of resonators while $\D_\frm$ is a finite lattice of resonators of width $r$. For $i=1,...,N$ and $m\in \Lambda$, we let $D_i^m$ denote the $i$\textsuperscript{th} resonator inside the $m$\textsuperscript{th} cell:
$$D_i^m = D_i + m.$$

Next, we will define the capacitance coefficients associated to $\D$ and $\D_\frm$, starting with the finite structure $\D_\frm$. Let $G$ be the Green's function for Laplace's equation in three dimensions:
$$G(x) = -\frac{1}{4\pi|x|}.$$
Given a smooth, bounded domain $\Omega \subset \R^3$, the \emph{single layer potential}  $\mathcal{S}_\Omega: L^2(\p \Omega) \to H^1(\p \Omega)$ is defined as
$$\mathcal{S}_\Omega[\varphi](x) := \int_{\partial \Omega} G(x-y) \varphi(y) \dx\sigma(y),\quad x\in \p \Omega.$$
Crucially for the analysis that will follow, $\S_\Omega$ is known to be invertible \cite{ammari2009layer}. For the finite lattice $\D_\frm$, we define the capacitance coefficients as 
\begin{equation} \label{eq:Cfinite}
	C^{mn}_{\frm,ij}(r) = \int_{\p D_i^m} \S_{\D_\frm}^{-1}[\chi_{\p D_j^n}]  \dx \sigma,
\end{equation}
for $1\leq i,j\leq N$ and $m,n \in I_r$, where $\chi_A(x)$ denotes the indicator function of the set $A\subset \R^3$. Here, we explicitly indicate the dependence of the size $r$ of the truncated lattice. For $m,n \in I_r$, we observe that $C^{mn}_\frm(r)$ is a matrix of size $N\times N$, while the block matrix $C_\frm = (C^{mn}_\frm)$ is a matrix of size  $N|I_r|\times N|I_r|$.

We can define analogous capacitance coefficients for the infinite structure $\D$. We begin by defining the dual lattice $\Lambda^*$ of $\Lambda$ as the lattice generated by the dual lattice vectors $\hat{\alpha}_1,...,\hat{\alpha}_{d}$ satisfying $ \hat{\alpha}_i\cdot l_j = 2\pi \delta_{ij}$  for $i,j = 1,...,d$ and whose projection onto the orthogonal complement of $\Lambda$ vanishes. We define the \emph{Brillouin zone} $Y^*$ as $Y^*:= \big(\R^{d}\times\{\mathbf{0}\}\big) / \Lambda^*$, where $\mathbf{0}$ is the zero-vector in $\R^{3-d}$. We remark that $Y^*$ can be written as $Y^*=Y^*_d\times\{\mathbf{0}\}$, where  $Y^*_d$ has the topology of a torus in $d$ dimensions.

When $\alpha\in Y^*\setminus \{0\}$, we can define the quasi-periodic Green's function $G^{\alpha}(x)$ as
\begin{equation}\label{eq:xrep}
	G^{\alpha}(x) := \sum_{m \in \Lambda} G(x-m)e^{\iu \alpha \cdot m}.
\end{equation}
The series in \eqref{eq:xrep} converges uniformly for $x$ and $y$ in compact sets of $\R^d$, with $x\neq y$ and $\alpha \neq 0$. Given a bounded domain $\Omega \subset Y$, the \emph{quasi-periodic} single layer potential  $\mathcal{S}_\Omega^{\alpha}: L^2(\p \Omega) \to H^1(\p \Omega)$ is then defined as
\begin{equation}
	\mathcal{S}_\Omega^{\alpha}[\varphi](x) := \int_{\partial \Omega} G^{\alpha} (x-y) \varphi(y) \dx\sigma(y),\quad x\in \p \Omega.
\end{equation}
For $\alpha \in Y^*$ and for $1\leq i,j\leq N$, we have a ``dual-space'' representation of the infinite capacitance matrix as the $N\times N$-matrix
\begin{equation} \label{eq:Calpha}
	\widehat{C}_{ij}^\alpha = \int_{\p D_i} (\S_D^\alpha)^{-1}[\chi_{\p D_j}] \dx \sigma.
\end{equation}
This is a ``dual-space'' representation in the sense that it is parametrised by $\alpha$ which is the Floquet-Bloch parameter that describes the frequency of spatial oscillation of the eigenmodes. Thus, we can equivalently define a ``real-space'' representation of the capacitance coefficients through an appropriate transformation. That is, for $1\leq i,j\leq N$, the ``real-space'' capacitance coefficients at the lattice point $m$ are given by
\begin{equation} \label{eq:Crealspace}
	C_{ij}^m = \frac{1}{|Y^*|}\int_{Y^*} \widehat{C}_{ij}^\alpha e^{-\iu \alpha\cdot m}\dx \alpha.
\end{equation}
Here, $C_{ij}^0$ corresponds to the diagonal block which contains the capacitance coefficients of the resonators within a single unit cell. We use the notation $\Cf$ to denote the infinite matrix that contains all the $C_{ij}^m$ coefficients, for all $1\leq i,j\leq N$ and all $m\in\Lambda$.

The main results in this work are based on relating the finite capacitance matrix $C_\frm$ to the \emph{truncated} capacitance matrix $C_\trm$. This matrix is obtained by truncating $\Cf$ to the centre block of size $N|I_r|\times N|I_r|$, to give a matrix of the same dimensions as $C_\frm$. The main technical result of this work is \Cref{lem:ae}, which ascertains a type of weak convergence of $C_\frm$ to $C_\trm$ as $r\to \infty$. 

The spectra of the infinite structure and the finite structure, respectively, are given by the solutions $\omega$ to the spectral problems
\begin{equation}
	\Cf \uf = \omega^2 \uf \quad \text{and} \quad C_\frm u = \omega^2 u.
\end{equation}
The goal of this work is to compare spectral properties of the infinite structure and the finite structure. Specifically, this work will focus on the convergence of eigenvalues of $C_\frm$ to the \emph{essential} spectrum of $\Cf$; the convergence of pure-point spectra  (defect modes) has  already been treated in \cite{ammari2023defect}. Throughout, we let $\widehat{\omega}_k(\alpha),$ for $k=1,...,N$ denote the positive eigenvalues of the quasi-periodic capacitance matrix problem 
$$\widehat{C}^\alpha u = \left(\widehat{\omega}_k(\alpha)\right)^2u, \quad k = 1,...,N,$$ 
and let $\omega_i$, for $i=1,...,N|I_r|$ denote the positive eigenvalues of the finite capacitance matrix problem
$$C_\frm u = \left(\omega_i\right)^2u, \quad i = 1,...,N|I_r|.$$

We conclude this section with the following convergence result of the capacitance coefficients, which was proved in \cite{ammari2023defect}.
\begin{thm} \label{thm:approx}
	For fixed $m,n \in \Lambda$, we have as $r\to \infty$, 
	$$\lim_{r\to \infty} C^{mn}_\frm(r) = C^{m-n},$$
	where $C^{mn}_\frm = (C^{mn}_{\frm,ij})_{i,j=1}^N$ and $C^{m-n} = (C^{m-n}_{ij})_{i,j=1}^N$ denote, respectively, the $N\times N$ matrices defined in \eqref{eq:Cfinite} and \eqref{eq:Crealspace}.
\end{thm}

\section{Eigenvalue distribution and essential spectral convergence} \label{sec:dos}
The main goal of this section is to prove the distributional convergence of the density of states of the finite and infinite materials. The main technical result is \Cref{lem:ae}, which establishes the convergence of the finite and truncated capacitance matrices in a certain weak ``averaged'' norm. We emphasise that the finite and truncated capacitance matrices are not expected to converge strongly in the matrix operator norm. This is due to the fact that, regardless of its size, the finite structure will always exhibit edge effects.

\subsection{Density of states}
For an $N$-level system with band functions $\widehat{\omega}_k(\alpha)$, $k=1,...,N$, we define the ``finite-material'' density of states $D_\frm(\omega)$ and ``infinite-material'' denisty of states $D(\omega)$ as the distributions (see, for example, \cite{ziman1972principles,economou2006green})
\begin{equation}\label{eq:D}
	D_\frm(\omega) = \frac{1}{N|I_r|}\sum_{i=1}^{N|I_r|}\delta(\omega - \omega_i) \quad \text{and} \quad
D(\omega) = \frac{1}{(2\pi)^d}\int_{Y^*}\frac{1}{N}\sum_{k=1}^N \delta\bigl(\omega - \widehat{\omega}_k(\alpha) \bigr)\dx \alpha.
\end{equation}
Let $L(\omega) = \{\alpha \in Y^* \mid \widehat{\omega}_k(\alpha) = \omega \text{ for some } k\}$ be the level set of $\widehat{\omega}_k(\alpha)$ at $\omega$. Carrying out the integral in \eqref{eq:D}, we can rewrite $D$ as 
\begin{equation} \label{eq:dos}
	D(\omega) = \frac{1}{(2\pi)^d}\int_{L(\omega)}\frac{1}{N}\sum_{k=1}^N \frac{1}{|\nabla_\alpha\widehat{\omega}_k(\alpha)|}\dx \sigma(\alpha),
\end{equation}
where $\dx \sigma(\alpha)$ is the surface measure on $L(\omega)$. Points $\alpha$ where $|\nabla_\alpha\widehat{\omega}_k(\alpha)|=0$ are known as \emph{van Hove singularities} and occur around any band edge \cite{van1953occurrence}. 

\subsection{Finite-material eigenvalue distribution}
The truncated matrix $C_\trm$ of size $m\in \Z^d$ is a multilevel block Toeplitz matrices, with known asymptotic eigenvalue distribution in terms of the eigenvalues $\widehat{\omega}_k(\alpha)$ of the quasi-periodic capacitance matrix $C^\alpha$ (see, e.g. \cite{gazzah2001asymptotic,gray1972asymptotic} for the one-dimensional case and  \cite{tyrtyshnikov1998spectra,tilli1998note} for the two- and three-dimensional cases). Based on these results, we will show that the finite capacitance matrix $C_\frm$ has identical eigenvalue distribution to $C_\trm$ as the size tends to infinity. 	

For an $n\times n$  matrix $M$, let $|M|$ denote the normalized Frobenius norm
\begin{equation}
|M|^2 = \frac{1}{n}\sum_{i,j=1}^n|m_{i,j}|^2.
\end{equation}
We will use $\|M\|_2$ for the standard Euclidean matrix norm. The following lemma is the main technical result we will need.
\begin{lem}\label{lem:ae}
	As $r\to \infty$, the matrices $C_\trm$ and $C_\frm$ are asymptotically equivalent, in other words, it holds that
	\begin{itemize}
		\item $\ds \lim_{r\to \infty}|C_\frm - C_\trm| = 0$;
		\item $\|C_\frm\|_2$ and $\|C_\trm\|_2$ are uniformly bounded as $r\to \infty$.
	\end{itemize}
	
\end{lem}

\begin{proof}
	$\|C_\frm\|_2$ is uniformly bounded by \cite[Lemma 3.4]{ammari2023defect}, while $\|C_\trm\|_2$ is uniformly bounded since it is the Toeplitz matrix of an essentially bounded symbol.
	
	Let $\D_\frm$ denote the finite lattice of width $r$. \Cref{thm:approx} tells us how to extend this finite lattice $\D_\frm$ to a larger lattice $\tilde{\D}_\frm$, which has width $\tilde{r} > r$. In particular, \Cref{thm:approx} shows that we can make this extension in such a way that the corresponding ``$r$-sized'' block $\tilde{C}_{\frm,0}$ of the finite capacitance matrix corresponding to $\tilde{\D}_\frm$ is arbitrarily close to the ``$r$-sized'' truncated matrix $C_\trm$. That is, given an $\epsilon>0$, we can make a sufficiently large extension such that
	\begin{equation}\label{eq:cf0}
		\|C_\trm - \tilde{C}_{\frm,0}\|_2 < \epsilon.	
	\end{equation}
	Observe that  
	\begin{equation}\label{eq:Cdiff}
		\left| C_\frm - \tilde{C}_{\frm,0}\right|^2 = \frac{1}{N|I_r|}\sum_{\substack{m,n\in I_r\\ 1\leq i,j\leq N}} \left( \int_{\p D^m_i} \left(\S_{\D_f}^{-1} - \S_{\tilde{\D}_f}^{-1} \right)[\chi_{\p D^n_j}]\dx \sigma\right)^2.
	\end{equation}
	We define the ``tail'' $\D_0$ of the extended lattice as
	$$\D_0 = \tilde{\D}_\frm \setminus \D_\frm.$$
	Observe that we have a block-structure of the single-layer potential on the extended lattice:
	$$\S_{\tilde \D_\frm} = \begin{pmatrix} \S_{\D_f} & \S_{\D_0}|_{\D_\frm} \\ \S_{\D_\frm}|_{\D_0} &  \S_{\D_0}\end{pmatrix}.$$
	We then have a block inverse of $\S_{\tilde \D_\frm}$ as follows:
	\begin{equation}\label{eq:Siblock}
	\S_{\tilde \D_\frm}^{-1} =  \begin{pmatrix} \left(\S_{\D_\frm} - \S_{\D_0}|_{\D_\frm} \S_{\D_0}^{-1} \S_{\D_\frm}|_{\D_0}\right)^{-1}& A_1\\ A_2 &  A_3\end{pmatrix},
	\end{equation}
	where $A_i$ are bounded operators that are immaterial for our analysis. We are now ready to start estimating the difference between $C_\frm$ and $C_\trm$. We want to estimate the term
	$$\S_{\D_0}|_{\D_\frm} \S_{\D_0}^{-1} \S_{\D_\frm}|_{\D_0} \S_{\D_\frm}^{-1}[\chi_{\p D^n_j}].$$
	Define $U(x) = \S_{\D_\frm}\S_{\D_\frm}^{-1}[\chi_{\p D^n_j}]$ and $V = \S_{\D_0}\S_{\D_0}^{-1}[U|_{\p \D_0}]$; these functions satisfy the systems of equations 
	\begin{equation}
	\begin{cases}
		\Delta U = 0, \quad x\in \R^3 \setminus \D_\frm, \\
		U|_{\p \D_\frm} = \chi_{\p D^m_{i}},\\
		U(x) \sim \frac{1}{|x|},
	\end{cases}
	\quad\text{and}\quad
	\begin{cases}
		\Delta V = 0, \quad &x\in \R^3 \setminus \D_0, \\
		V(x) = U(x), & x\in \p \D_0, \\
		V(x) \sim \frac{1}{|x|}.
	\end{cases}
	\end{equation}
	Observe that the boundary conditions satisfied by $U$ are imposed on $\p \D_\frm$ while the boundary conditions satisfied by $V$ are imposed on $\p \D_0$. In particular, $U(x)$ scales like $|m|^{-1}$ for $x\in \p \D_0$ while $V(x)$ scales like $(|m||n|)^{-1}$ for $x \in \p D^n$. As $r\to \infty$, we therefore have 
	$$\int_{\p D^m_i} \left(\S_{\D_f}^{-1} - \S_{\tilde{\D}_f}^{-1} \right) [\chi_{\p D^n_j}]\dx \sigma \leq \frac{K_1}{(1+|m|)(1+|n|)},$$
	for some constant $K_1$. From \eqref{eq:Cdiff} and \eqref{eq:Siblock}, we use the Neumann series to find that 
	\begin{align*}
		\left| C_\frm - \tilde{C}_{\frm,0}\right|^2 &\leq  \frac{K_2}{r^d}\sum_{m,n\in I_r} \frac{1}{(1+|m|)^2(1+|n|)^2} \\
		&\leq  K_3r^{d-4},
	\end{align*}
	where $d\in \{1,2,3\}$. In other words, $\left| C_\frm - \tilde{C}_{\frm,0}\right| \to 0,$  which together with \eqref{eq:cf0} concludes the proof.
\end{proof}

From \cite{gray1972asymptotic} we know that asymptotically equivalent matrices have identical eigenvalue distributions as their sizes tend to infinity. This gives the following result on distributional convergence of the discrete density of states $D_\frm(\omega)$ to the continuous density of states $D(\omega)$.
\begin{thm}\label{thm:dist}
	As $r\to \infty$, $D_\frm(\omega)$ converges to $D(\omega)$ in the sense of distributions. In other words, for any smooth function $F$ with compact support, we have 
	$$\lim_{r\to \infty}\int_{-\infty}^{\infty}  D_\frm(\omega) F(\omega) \dx \omega = \int_{-\infty}^\infty D(\omega) F(\omega)  \dx \omega.$$
\end{thm}
\begin{proof}
	Since $C_\trm$ and $C_\frm$ are asymptotically equivalent, we have from \cite{gray1972asymptotic} that 
	$$\lim_{r\to \infty} \frac{1}{N|I_r|}\sum_{i=1}^{N|I_r|} F(\omega_i) = \frac{1}{(2\pi)^d} \int_{Y^*}\frac{1}{N}\sum_{k=1}^N F\bigl(\widehat{\omega}_k(\alpha)\bigr)\dx \alpha,$$
	from which the theorem follows.
\end{proof}

From \Cref{thm:dist} we have that the frequencies $\omega_{i}, i =1,...,N|I_r|,$ of the finite capacitance matrix are distributed according to 
$$\omega_{i} \sim \widehat{\omega}_{k}(\alpha),$$
where $\alpha$ is uniformly distributed on the Brillouin zone $Y^*$ and $\widehat{\omega}_{k}(\alpha), k = 1,...,N,$ are the eigenvalues of the quasi-periodic capacitance matrix. The proportion of modes with eigenfrequencies in the infinitesimal interval between $\omega$ and $\omega+\dx \omega$ is then approximated by
$$
	D_\frm(\omega) \dx \omega \approx \frac{\dx \omega }{(2\pi)^d}\int_{L(\omega)} \frac{1}{N}\sum_{k=1}^N\frac{1}{|\nabla_\alpha\omega_k(\alpha)|}\dx \sigma(\alpha).
$$ 
For a one-dimensional chain of single resonators ($d=1, N=1$) we obtain
\begin{equation} \label{eq:DOS1d}
D_\frm(\omega) \dx \omega \approx \frac{\dx \omega }{2\pi}\frac{1}{|\widehat{\omega}'(\alpha)|}\Bigg|_{\alpha \in \alpha(\omega)},
\end{equation}
where $\alpha(\omega) = \{\alpha\in Y^* \mid \widehat{\omega}_1(\alpha) = \omega\}$, which is shown by the solid lines in \Cref{fig:hists1d}.

\subsection{Numerical results}

The convergence of the distribution of the resonant frequencies can be studied numerically. In \Cref{fig:convergence1d} we plot histograms of the discrete, finite set of subwavelength resonant frequencies for truncated structures and compare this to the density of states (DOS) for the infinite array, given by \eqref{eq:DOS1d}. In each case, the histograms and the DOS are normalised so that the area under the curves is equal to 1. We can see that the distribution of the truncated eigenvalues closely resembles the DOS as the structure becomes sufficiently large (for 1000 resonators, the curve is difficult to distinguish from the histogram in our plot). We can quantify the error by computing the area between the two curves (the histograms being viewed as curves for these purposes). This is shown in the lower plot of \Cref{fig:convergence1d} and we observe linear convergence as the size of the array increases. In other words, the discrete DOS of the truncated resonant frequencies is converging in distribution to the DOS, at a linear rate.

\begin{figure}
	\centering
	\begin{subfigure}{\linewidth}
		\centering
		\includegraphics[width=\linewidth]{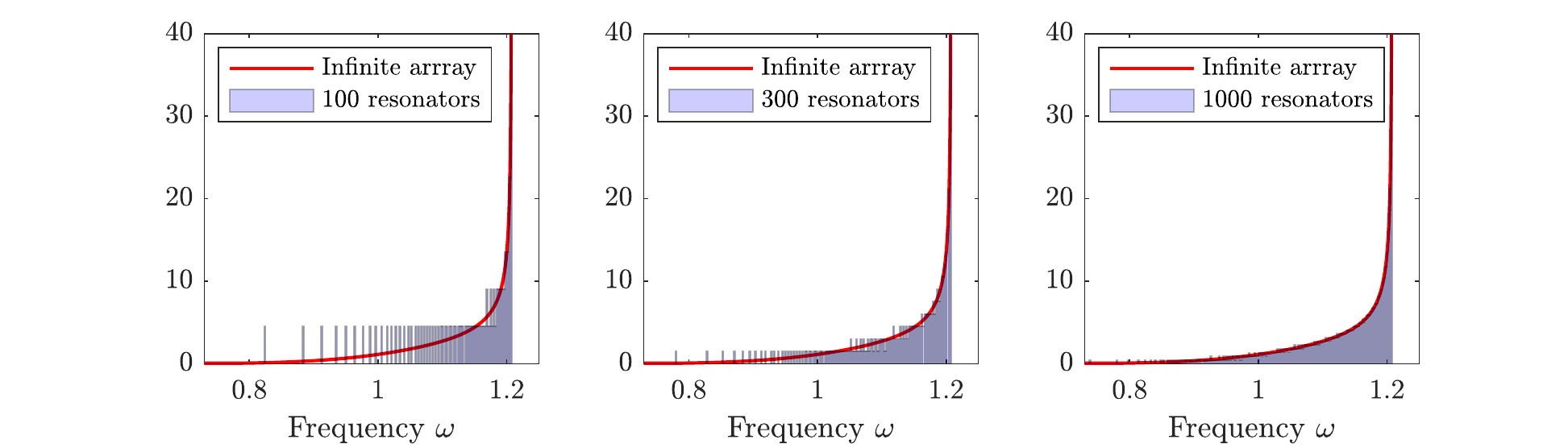}
		\caption{} \label{fig:hists1d}
	\end{subfigure}
	
	\vspace{0.3cm}
	
	\begin{subfigure}{\linewidth}
		\centering
		\includegraphics[width=0.65\linewidth]{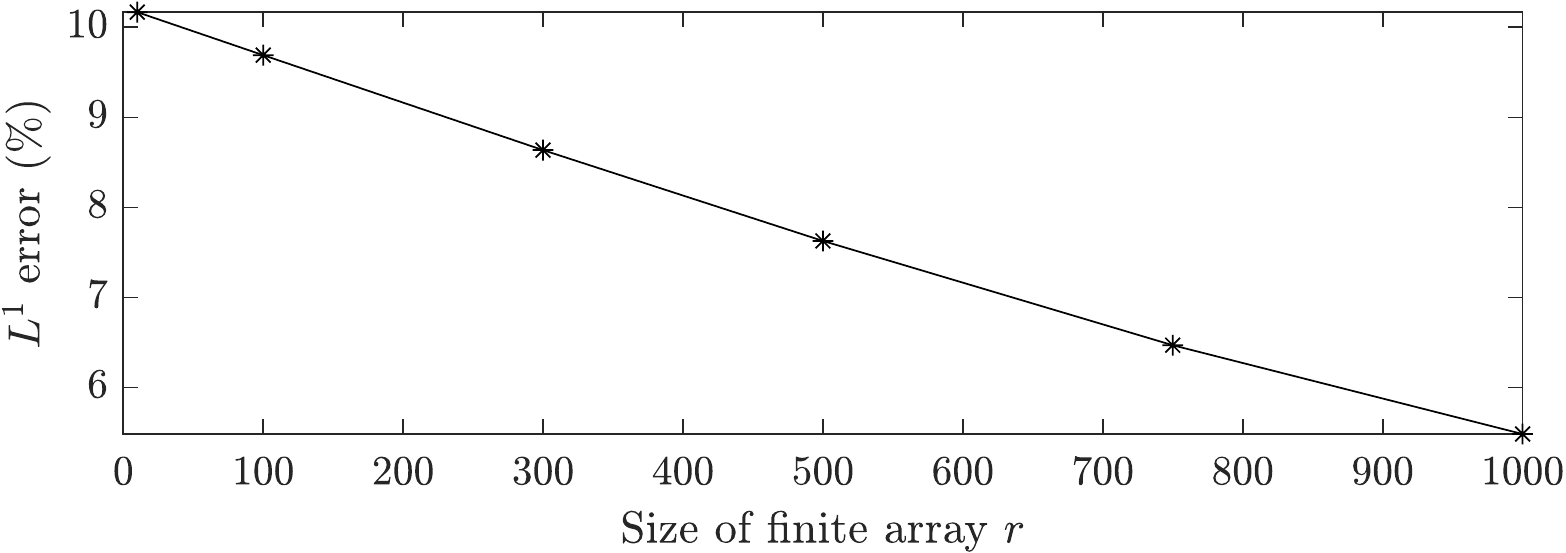}
		\caption{}
	\end{subfigure}
	\caption{Convergence in distribution of the resonant frequencies of the truncated linear array to the density of states (DOS) for the infinite array. (a) The resonant frequencies of the truncated arrays are shown in histograms and the DOS for the infinite array \eqref{eq:dos} as a solid red line. Both the histograms and the DOS have been normalised to have unit area under the curves. (b) The $L^1$ error between the histogram plots and the DOS curve, which converges to zero as the size of the finite array increases.} \label{fig:convergence1d}
\end{figure}

Similar histograms can be produced for multi-dimensional lattices. For example, in \Cref{fig:convergence2d} we plot the same histograms for successively larger square (two-dimensional) lattices. Once again, we can see that the distribution of the eigenfrequencies converges to a fixed distribution as the size of the finite lattice increases. One notable difference from the one-dimensional lattice shown in \Cref{fig:convergence1d} is that the distribution is not singular at the edge of the first band. 

\begin{figure}
	\centering
	\includegraphics[width=\linewidth]{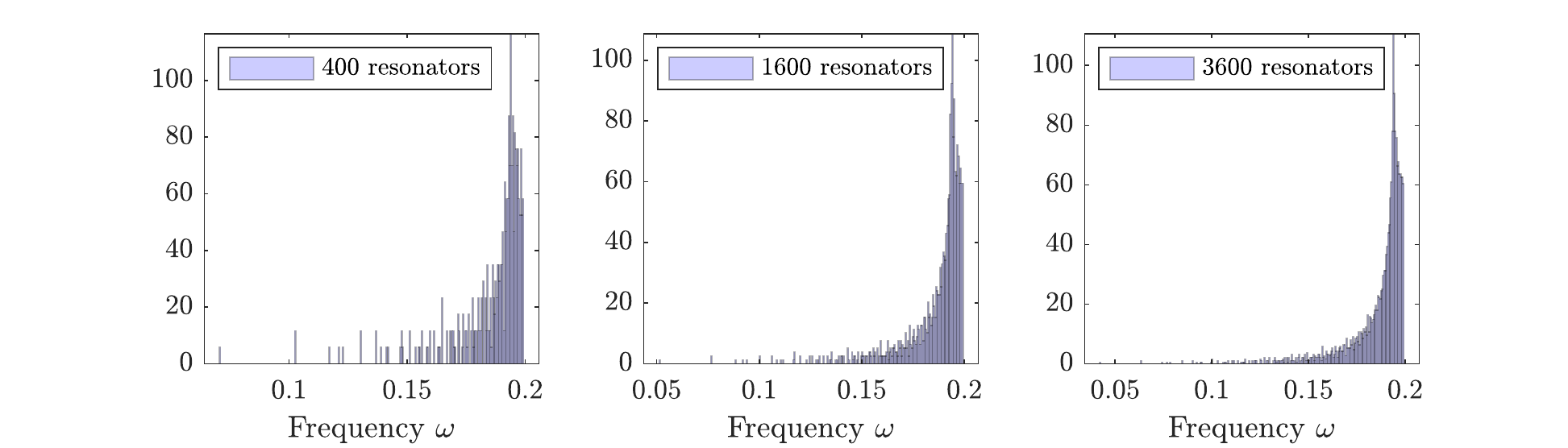}
	\caption{Distribution of the resonant frequencies of a truncated square (two-dimensional) array, plotted as histograms.} \label{fig:convergence2d}
\end{figure}

\section{Pointwise convergence and discrete band structure} \label{sec:bands}
Since any edge effects persist in the limit $r\to\infty$, we should not expect all eigenvalues of $C_\frm$ to converge to the spectrum of $\Cf$.  Nevertheless, for any point in the continuous spectrum, there will always be eigenvalues arbitrarily close. We can repeat the arguments used in the proof of \Cref{lem:ae} to obtain the following theorem on pointwise convergence of the essential spectrum and Bloch modes.
\begin{thm}\label{thm:pwconv}
	Let $\omega = \widehat{\omega}_k(\alpha)$ be an eigenvalue of the quasi-periodic capacitance matrix $C^\alpha$ for some $\alpha \in Y^*$, corresponding to the normalized eigenvector $u\in \R^N$. Then, for any $\epsilon>0$ we can choose $r>0$ such that the finite capacitance matrix $C_\frm(r)$ has a family of eigenvalues $\omega_i^2, i \in I$ and associated normalized eigenvectors $u_i \in \R^{N|I_r|}$ satisfying
	$$|\omega^2-\omega_i^2| < \epsilon, \qquad \left\|\tilde{u}-\sum_{i\in I}c_iu_i\right\|_2 < \epsilon,$$
	for some $c_i \in \Cb$, 	where $\tilde{u}\in \R^{N|I_r|}$ is the normalized quasi-periodic extension of $u$.
\end{thm}

\begin{rmk}
	Although any eigenvalue and eigenvector of the quasi-periodic capacitance matrix can be approximated by eigenvalues and eigenvectors of the finite capacitance matrix, the converse need not hold. Indeed, due to edge effects, $C_\frm$ might have eigenvalues which do not approach those of $C_\trm$ in the limit $r\to \infty$. Although $C_\frm$ converges to $C_\trm$ in the (weak) norm $| \cdot |$, it does not converge in the (strong) Euclidean operator norm.  
\end{rmk}

The discrete band function calculation introduced in \cite{ammari2023defect} provides a notion of how well an eigenmode of $C_\frm$, for large $r$, is approximated by Bloch modes of the infinite structure. Given an eigenmode $u_j$, we can take the truncated Floquet transform  of $u_j$ as
\begin{equation} \label{def:truncFloquet}
	(\widehat{u}_j)_\alpha = \sum_{m\in I_r} (u_j)_me^{\iu \alpha\cdot m}, \qquad \alpha \in  Y^*.
\end{equation}
Here we denote by $(u_j)_m$ the vector of length $N$ associated to cell $m\in \Lambda$. Observe that $u_j$ is a vector of length $N|I_r|$ while $(\widehat{u}_j)_\alpha $ is a vector of length $N$. Looking at the Euclidean 2-norm $\|(\widehat{u}_j)_\alpha\|_2$ as a function of $\alpha$, this function has distinct peaks, which are the quasi-periodicities $\alpha$ associated with $u_j$. An example is shown in \Cref{fig:alpha}. We can then define a discrete band structure whereby the eigenvalues $\omega_j$ of $C_\frm$ are associated to a quasi-periodicity $\alpha_j$ given by
\begin{equation} \label{eq:quasip}
	\alpha_j = \argmax_{\alpha \in Y^*} \|(\widehat{u}_j)_\alpha\|_2.
\end{equation}
Note that the symmetry of the problem means that if $\alpha$ is an approximate quasi-periodicity then so will $-\alpha$ be. In cases of additional symmetries of the lattice, we expect additional symmetries of the quasi-periodicities. 

As a demonstrative example of this process, we consider the case of a single resonator repeating in one direction. This has a single continuous band of eigenfrequencies, as shown in \Cref{fig:band}. For a truncated version of the structure, the approximate band functions can be reconstructed as above. In \Cref{fig:alpha} we show the norm of the truncated Floquet transform of the 10\textsuperscript{th}, 20\textsuperscript{th} and 30\textsuperscript{th} eigenmodes in an array of 50 resonators. In each case the function is even about zero and has a clear peak, which allows us to identify an appropriate quasi-periodicity $\alpha_j$ via \eqref{eq:quasip}. These values can be used to plot the 50 resonant frequencies alongside the continuous bands of the limiting infinite structure, which is shown in \Cref{fig:band}. We see that, even for a set of 50 resonators, the approximate band function closely resembles that of the infinite structure.

\begin{figure}
\centering
\includegraphics[width=0.8\linewidth]{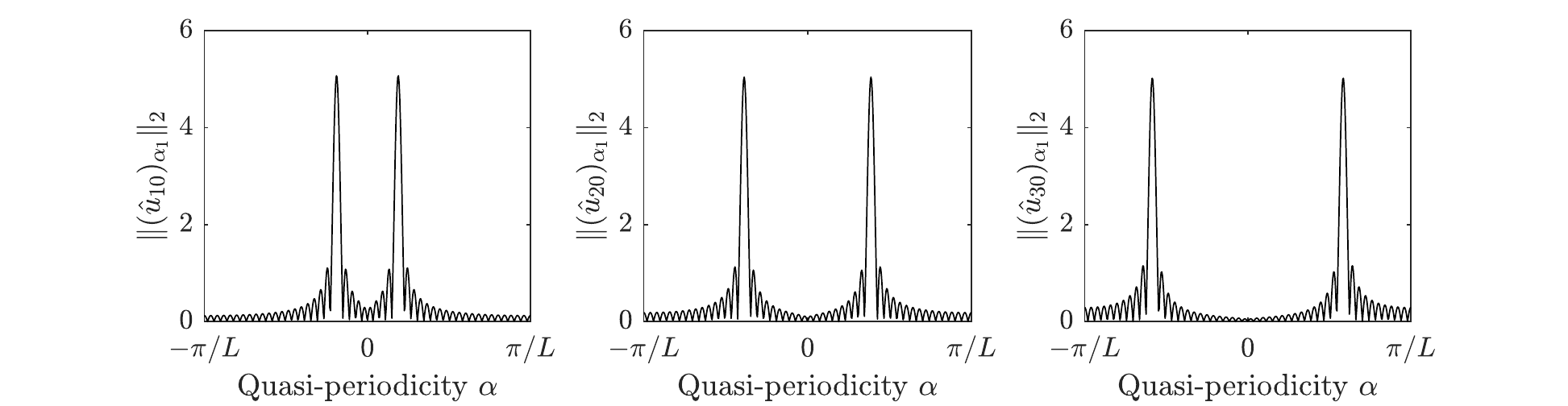}
\caption{A given eigenmode $u_j$ of a truncated periodic structure can be associated with a quasi-periodicity $\alpha$. Here, we plot the norm of the truncated Floquet transform of the 10\textsuperscript{th}, 20\textsuperscript{th} and 30\textsuperscript{th} eigenmodes of an array of 50 resonators. In each case, there are clear peaks that we can assign as the quasi-periodicities of the eigenmodes. These values can be used to reconstruct approximate spectral band functions.} \label{fig:alpha}
\end{figure}

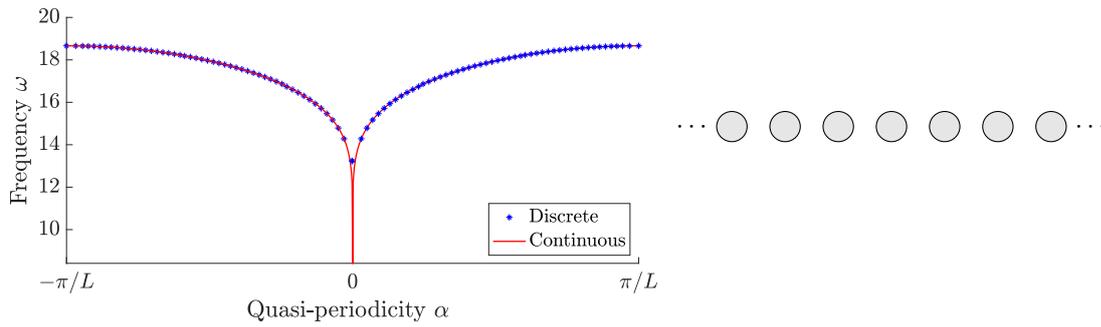
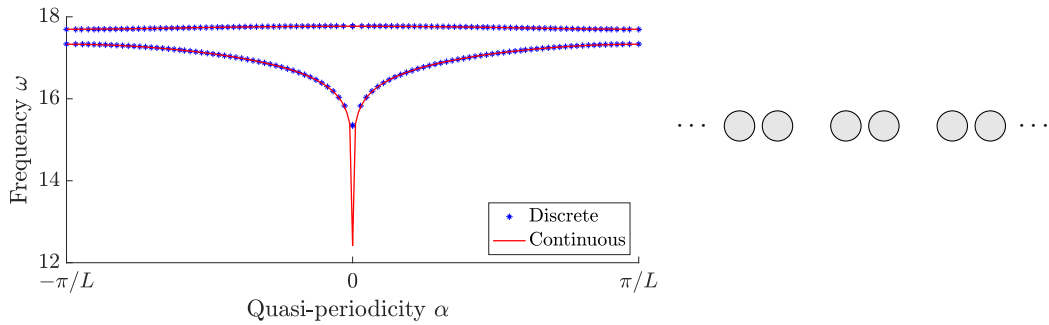
\begin{figure}
	\begin{subfigure}{\linewidth}
	\includegraphics[width=0.55\linewidth]{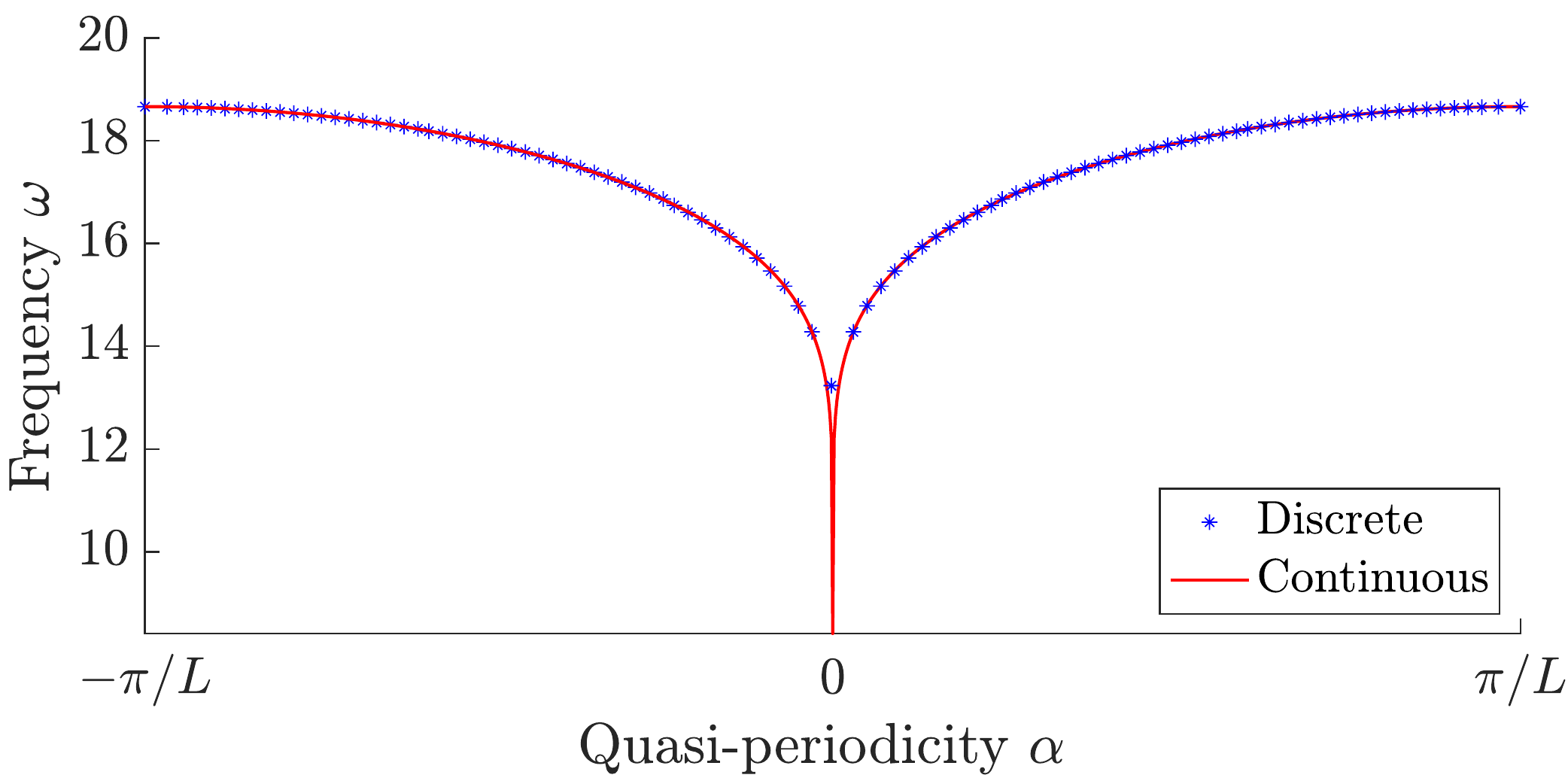}
	\begin{tikzpicture}
	\node at (0,0) {$\cdots$};
	\draw[fill=white!80!gray] (0.5,0) circle(0.2);
	\draw[fill=white!80!gray] (1.2,0) circle(0.2);
	\draw[fill=white!80!gray] (1.9,0) circle(0.2);
	\draw[fill=white!80!gray] (2.6,0) circle(0.2);
	\draw[fill=white!80!gray] (3.3,0) circle(0.2);
	\draw[fill=white!80!gray] (4,0) circle(0.2);
	\draw[fill=white!80!gray] (4.7,0) circle(0.2);
	\node at (5.25,0) {$\cdots$};
	\node[white] at (0,-2.5) {.};
	\end{tikzpicture}
	\caption{Single periodic resonators ($N=1$)} \label{fig:band}
	\end{subfigure}
	
	\vspace{0.2cm}
	
	\begin{subfigure}{\linewidth}
	\includegraphics[width=0.55\linewidth]{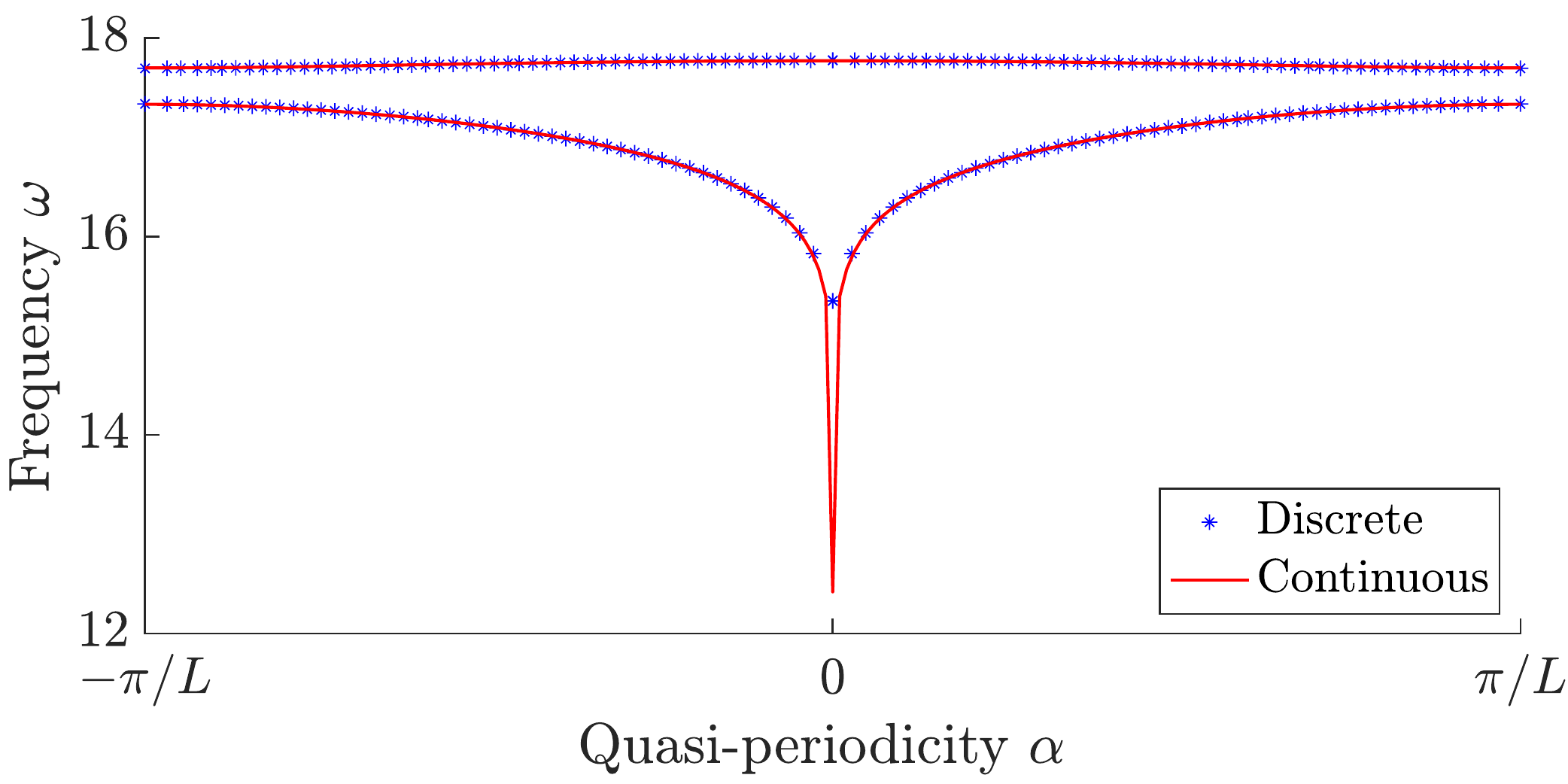}
	\begin{tikzpicture}
		\node at (0,0) {$\cdots$};
		\draw[fill=white!80!gray] (0.6,0) circle(0.2);
		\draw[fill=white!80!gray] (1.1,0) circle(0.2);
		\draw[fill=white!80!gray] (2,0) circle(0.2);
		\draw[fill=white!80!gray] (2.5,0) circle(0.2);
		\draw[fill=white!80!gray] (3.4,0) circle(0.2);
		\draw[fill=white!80!gray] (3.9,0) circle(0.2);
		\node at (4.5,0) {$\cdots$};
		\node[white] at (0,-2.5) {.};
	\end{tikzpicture}
	\caption{Periodic pairs of resonators ($N=2$)} \label{fig:band_dimer}
	\end{subfigure}
	\caption{The continuous spectrum of the infinite structure and the discrete spectrum of the truncated structure for a one-dimensional lattices. (a) Single periodic resonators ($N=1$) with a truncated structure consisting of 50 resonators. (b) Periodic pairs of resonators ($N=2$) with a truncated structure containing 100 resonators. In both cases, the truncated Floquet transform \eqref{def:truncFloquet} is used to approximate the quasi-periodicity of the truncated modes.} \label{fig:bands_1d}
\end{figure}

\begin{figure}
		\begin{subfigure}{\linewidth}
		\includegraphics[width=0.55\linewidth]{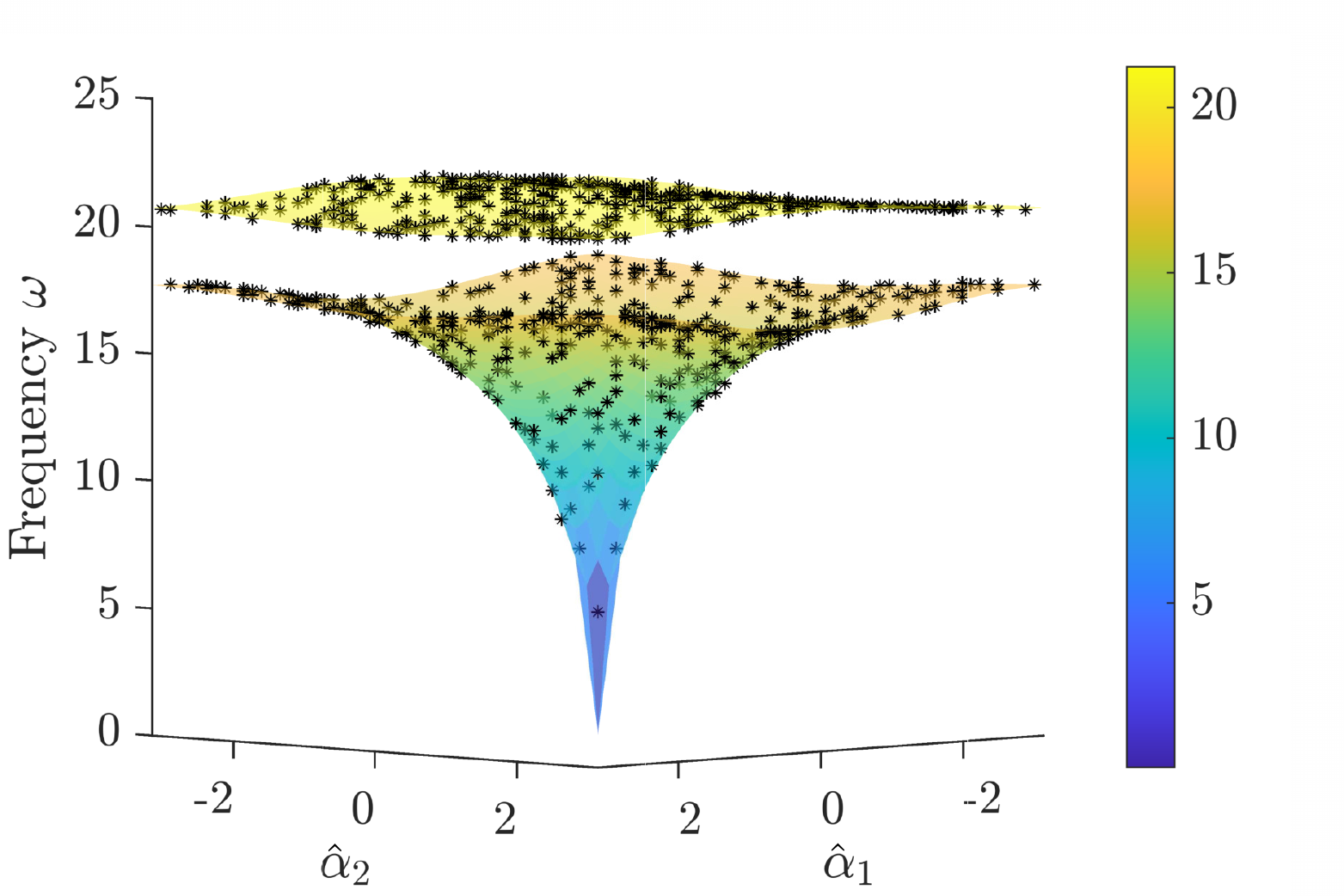}
		\begin{tikzpicture}
			\foreach \x in {1,2,...,5}{
				\foreach \y in {1,2,...,5}{
					\draw[fill=white!80!gray] (0.7*\x+0.15,0.7*\y) circle(0.12);
					\draw[fill=white!80!gray] (0.7*\x-0.15,0.7*\y) circle(0.12);
					\node at (0,0.7*\y) {$\cdots$};
					\node at (4.3,0.7*\y) {$\cdots$};
					\node at (0.7*\x,0.1) {$\vdots$};
					\node at (0.7*\x,4.3) {$\vdots$};
			}}
			\node[white] at (-1,-0.75) {.};
		\end{tikzpicture}
		\caption{Square lattice} \label{fig:band_square}
		\end{subfigure}
	\vspace{0.2cm}
	
	\begin{subfigure}{\linewidth}
		\includegraphics[width=0.55\linewidth]{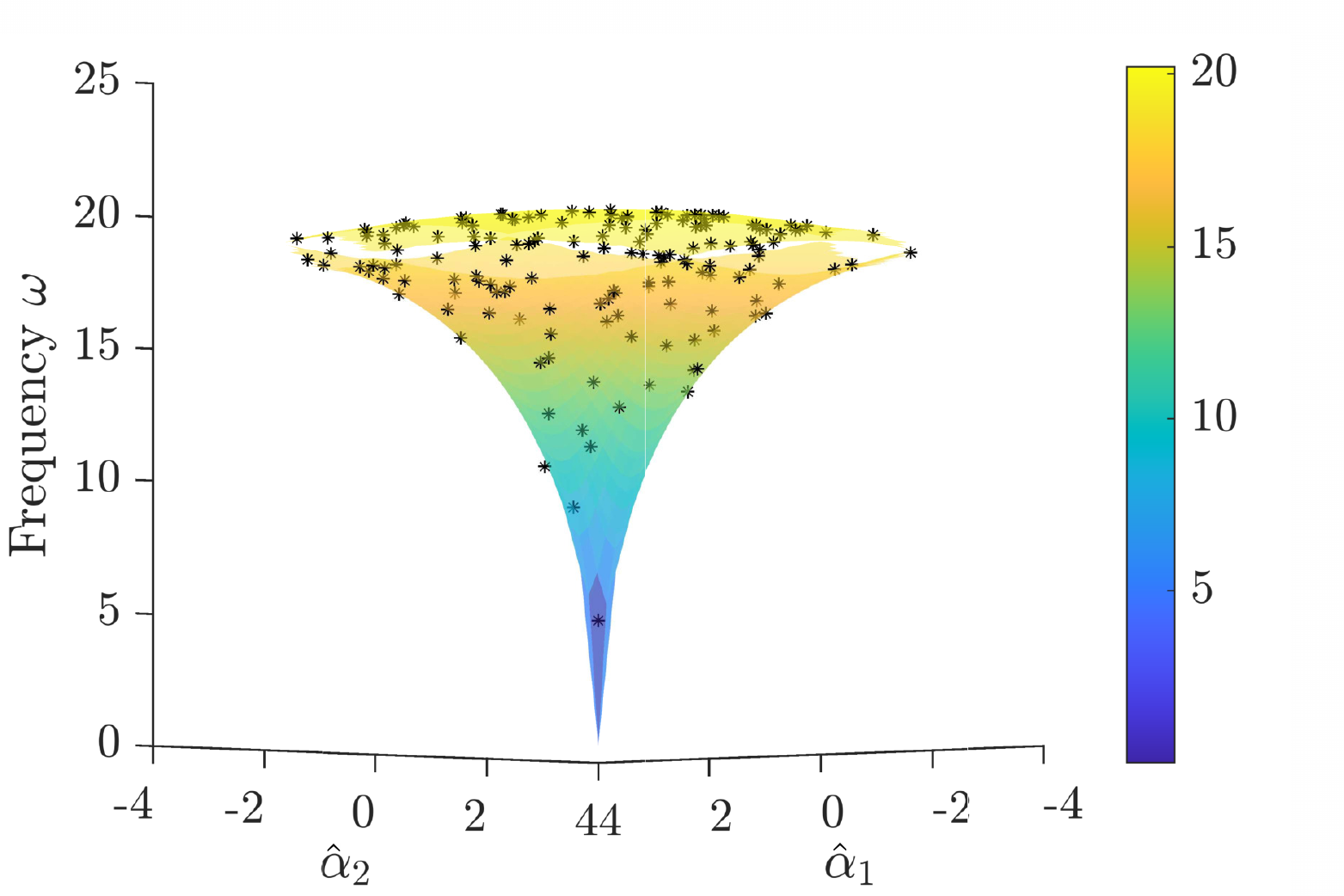}
		\begin{tikzpicture}[scale=0.8]
			\foreach \x in {1,2,...,4}{
				\foreach \y in {1,2,...,4}{
					\draw[fill=white!80!gray] (1.4*\x+0.7*\y,1.2124*\y) circle(0.17);
					\draw[fill=white!80!gray] (1.4*\x+0.7*\y+0.7,1.2124*\y+0.4041) circle(0.17);
					\node at (0.7*\y+0.9,1.2124*\y+0.25) {$\cdots$};
					\node at (0.7*\y+0.88+6,1.2124*\y+0.25) {$\cdots$};
					
					\node[rotate=60] at (0.7+1.4*\x,0.7) {$\cdots$};
					\node[rotate=60] at (3.5+1.4*\x,5.8) {$\cdots$};
			}}
			\node[white] at (1.6,-0.75) {.};
		\end{tikzpicture}		
		\caption{Honeycomb lattice} \label{fig:band_honeycomb}
	\end{subfigure}
	\caption{Examples of continuous and discrete spectra of the infinite and truncated structures, respectively. (a) A square lattice with two resonators per unit cell, resulting in two bands separated by a gap. (b) A honeycomb lattice with Dirac cones at the vertices of the Brillouin zone. In both cases, the truncated structures have 800 resonators and the truncated Floquet transform is used to approximate the quasi-periodicity of the truncated modes.} \label{fig:bands_2d}
\end{figure}

In \Cref{fig:band_dimer}, we compare the continuous and truncated spectra of an array of resonators arranged in pairs (dimers). The truncated structure has 100 resonators arranged in 50 pairs. This geometry is an example of the famous Su-Schrieffer-Heeger (SSH) chain \cite{SSH} which has been shown to have fascinating topological properties \cite{ammari2019topological}. This system has two subwavelength spectral bands and the truncated modes are split evenly between approximating the two bands.

Additionally, we can consider this method for lattices of higher dimension. \Cref{fig:band_square} shows the case of a square lattice of resonator dimers. Similarly to \Cref{fig:band_dimer}, there is a band gap between the first and the second bands, and we see a close agreement between the discrete and the continuous band structure. \Cref{fig:band_honeycomb} shows a similar figure in the case of a honeycomb lattice, where the finite lattice is truncated along zig-zag edges of the lattice. As shown in \cite{ammari2020honeycomb}, there are Dirac cones on each corner of the Brillouin zone. In the truncated structure, in addition to the ``bulk modes'' whose frequencies closely agree with the continuous spectrum, there are ``edge modes'' which are localized around the edges and whose points in the band structure lie away from the continuous bands.

\section{Nonperiodic band structure and topological invariants} \label{sec:nonperiodic}
The method described in \Cref{sec:bands} can also be applied to aperiodic structures that have been perturbed to introduce defects. Two examples of this are shown in \Cref{fig:bands_defects}. In each case, the defects have induced localised eigenmodes which do not have well-defined associated quasi-periodicities. These are shown with dashed lines. The rest of the truncated spectrum still agrees well with the continuous spectrum of the limiting infinite operator.

The example shown in \Cref{fig:pointdefect} is that of a local defect, where the material parameters are changed on the central resonator. As was studied in \cite{anderson, ammari2023defect}, this corresponds to multiplying the capacitance matrix by a diagonal matrix that is equal to the identity matrix other than a value greater than 1 in the central entry. A formula for the eigenfrequency of this defect mode in the infinite structure was derived in \cite{anderson} and it was proved in \cite{ammari2023defect} that the eigenfrequencies of the localised modes in the truncated arrays converges to that value.

In \Cref{fig:SSHdefect} we study the famous example of a interface mode in the Su-Schrieffer-Heeger (SSH) chain. This localised eigenmode exhibits enhanced robustness with respect to imperfections, a property it inherits from the underlying topological properties of the periodic structure via the notion of \emph{topological protection} \cite{ammari2019topological}. Even though the periodicity of the structure is broken due to the interface, we can still visualise the spectrum as a discrete band structure through \eqref{def:truncFloquet} and \eqref{eq:quasip}.

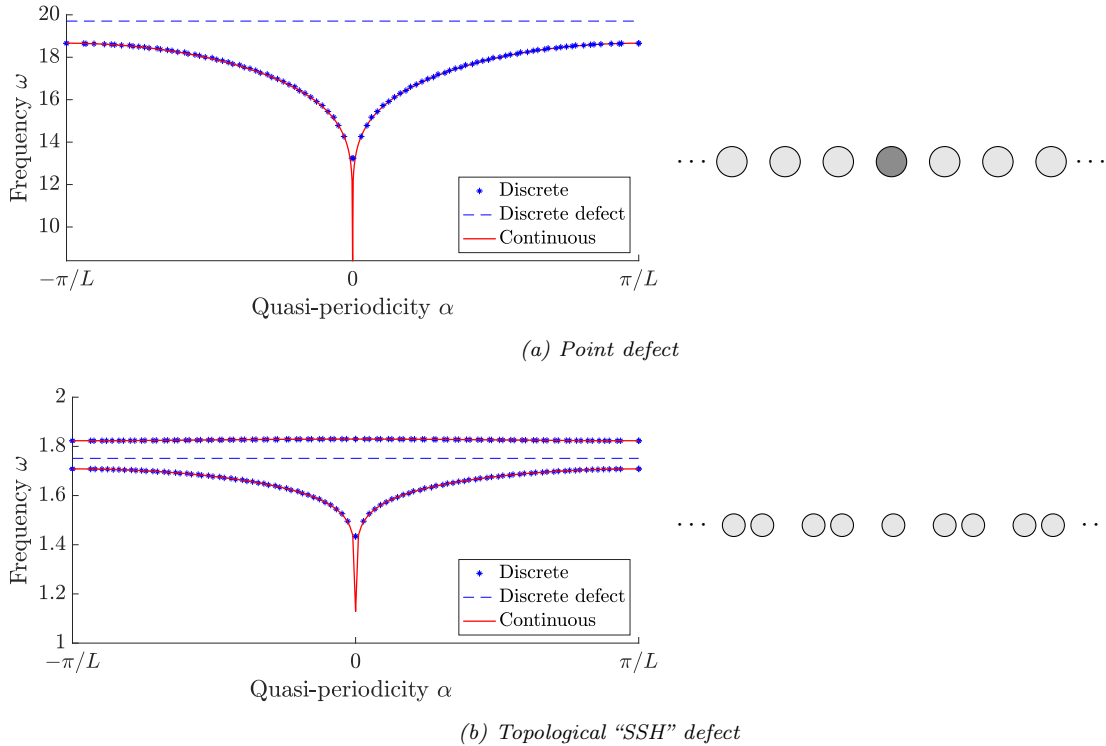
\begin{figure}
	\centering
	\begin{subfigure}{\linewidth}
	\includegraphics[width=0.55\linewidth]{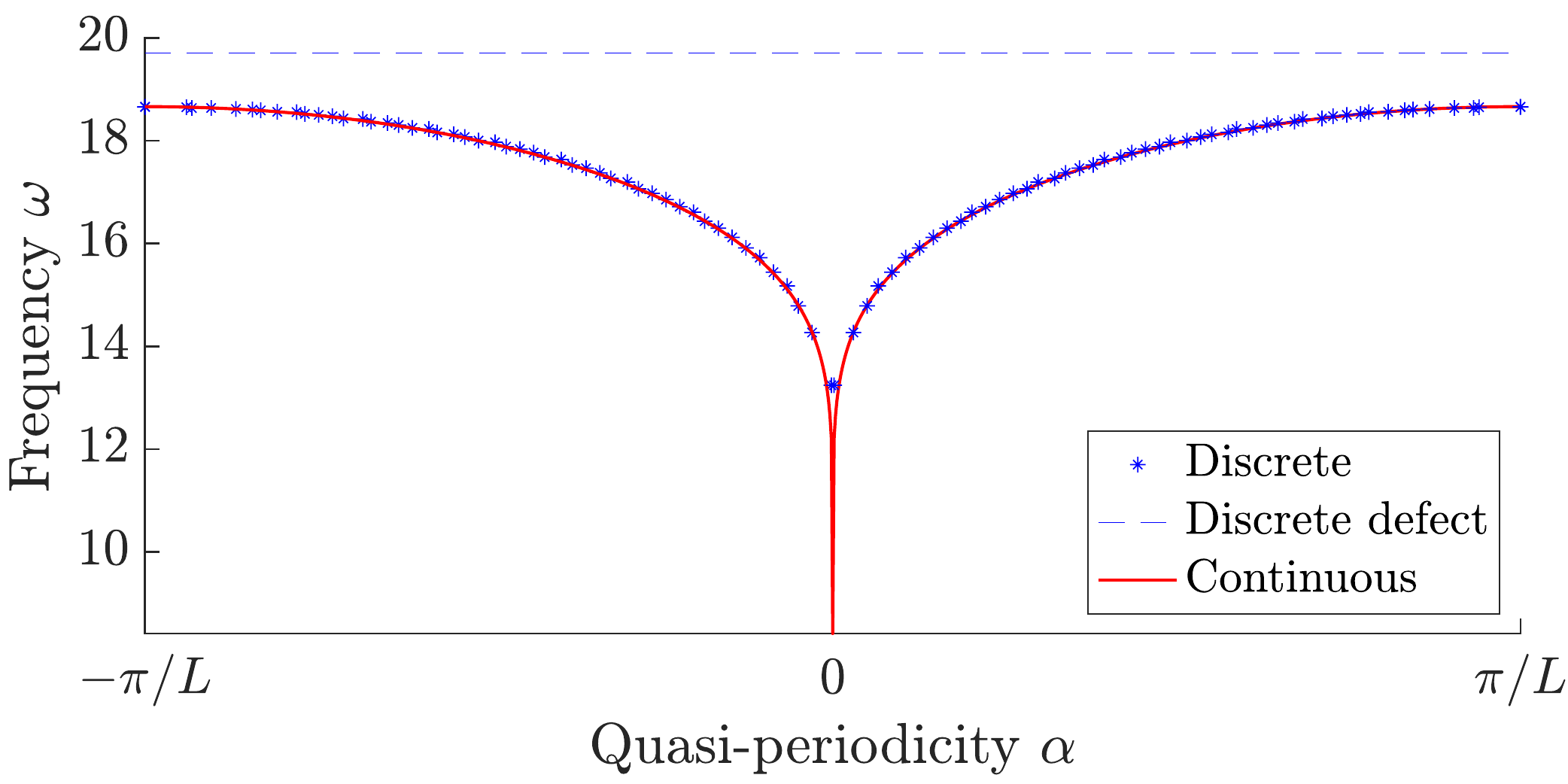}
	\begin{tikzpicture}
	\node at (0,0) {$\cdots$};
	\draw[fill=white!80!gray] (0.5,0) circle(0.2);
	\draw[fill=white!80!gray] (1.2,0) circle(0.2);
	\draw[fill=white!80!gray] (1.9,0) circle(0.2);
	\draw[fill=white!10!gray] (2.6,0) circle(0.2);
	\draw[fill=white!80!gray] (3.3,0) circle(0.2);
	\draw[fill=white!80!gray] (4,0) circle(0.2);
	\draw[fill=white!80!gray] (4.7,0) circle(0.2);
	\node at (5.25,0) {$\cdots$};
	\node[white] at (-0.1,-2) {.};
	\end{tikzpicture}
	\caption{Point defect} \label{fig:pointdefect}
	\end{subfigure}
	
	\vspace{0.2cm}
	
	\begin{subfigure}{\linewidth}
	\includegraphics[width=0.55\linewidth]{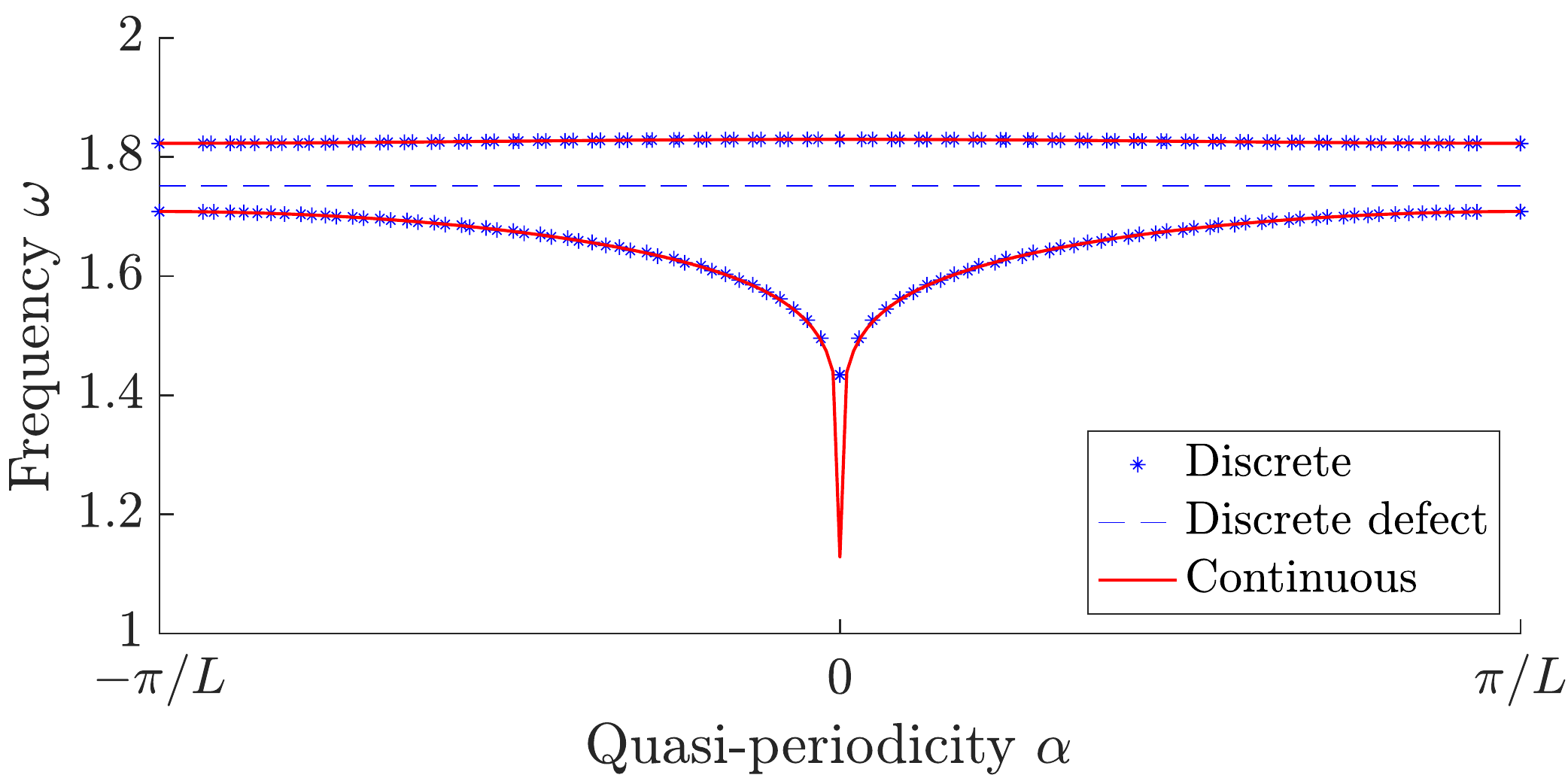}
	\begin{tikzpicture}[scale=0.75]
		\node at (-0.1,0) {$\cdots$};
		\draw[fill=white!80!gray] (0.6,0) circle(0.2);
		\draw[fill=white!80!gray] (1.1,0) circle(0.2);
		\draw[fill=white!80!gray] (2,0) circle(0.2);
		\draw[fill=white!80!gray] (2.5,0) circle(0.2);
		\draw[fill=white!80!gray] (3.4,0) circle(0.2);
		\draw[fill=white!80!gray] (4.3,0) circle(0.2);
		\draw[fill=white!80!gray] (4.8,0) circle(0.2);
		\draw[fill=white!80!gray] (5.7,0) circle(0.2);
		\draw[fill=white!80!gray] (6.2,0) circle(0.2);
		\node at (7,0) {$\cdots$};
		\node[white] at (0,-3) {.};
	\end{tikzpicture}
	\caption{Topological ``SSH'' defect} \label{fig:SSHdefect}
	\end{subfigure}
	\caption{The discrete spectrum a truncated array with defects can be related to the continuous spectrum of the unperturbed periodic structure. (a) A local defect in a periodic array of resonators. The truncated array has 51 resonators with the defect on the 26\textsuperscript{th}. The continuous spectrum of the infinite array is the same as that shown in \Cref{fig:band}. (b) A topological (non-compact) defect in an array of resonator pairs. The truncated array has 101 resonators with the defect on the 51\textsuperscript{st}. The continuous spectrum of the infinite array is the same as that shown in \Cref{fig:band_dimer}. In both cases, the defect introduces a localised eigenmode, which does not have a well-defined associated quasi-periodicity and is show with a dashed line.	 
%
%
	} \label{fig:bands_defects}
\end{figure}

\section{Concluding remarks}
In this work, we have shown the convergence of resonant frequencies of systems of coupled resonators in truncated periodic lattices to the essential spectrum of corresponding infinite lattice. We have studied this using the capacitance matrix model for coupled resonators with long-range interactions. We emphasise that our conclusions extends to other long-range models, since it is the decay of the coupling which is the main feature of the analysis.

The discrete band structure calculations in \Cref{sec:bands} give a concrete way to associate band structures to practically realizable materials which may be finite and aperiodic. Notably, for the field of topological insulators, this opens the possibility of defining \emph{discretely} defined invariants, defined only in terms of the eigenvalues and eigenmodes of the finite interface structures rather than the Bloch eigenvalues and eigenmodes of corresponding infinite, periodic, structures.

We emphasise that the matrix model adopted in this work is a Hermitian model with time-reversal symmetry. For non-Hermitian models, similar convergence theorems are not expected to hold, and the spectra of the finite and infinite models might be vastly different.  For the discrete band structure calculations in \Cref{sec:bands}, we expect the finite modes to be associated to a \emph{complex} momentum and to adequately describe these, we need to extend the Brillouin zone to the complex plane. A precise study of this setting would be a highly interesting future work.

\appendix
\section{Continuous PDE model} \label{app:cont}
In this appendix, we summarize how the generalized capacitance matrix gives an asymptotic characterisation a system of coupled high-contrast resonators. In particular, it can be used to characterize the subwavelength (\emph{i.e.} asymptotically low-frequency) resonance of the system. We refer the reader to \cite{ammari2021functional} for an in-depth review and extension to other settings. 

We will consider an array of finitely many resonators here, but the modification to infinite periodic systems is straightforward, though an appropriate modification of the Green's function \cite{ammari2021functional}. As previously considered, we suppose that the resonators are given by $D_i\subset\mathbb{R}^3$. We consider the scattering of time-harmonic waves with frequency $\omega$ and will solve a Helmholtz scattering problem in three dimensions. This Helmholtz problem, which can be used to model acoustic, elastic and polarized electromagnetic waves, represents the simplest model for wave propagation that still exhibits the rich phenomena associated to subwavelength physics.

We let $v_i$  denote the wave speed in each resonator $D_i$ so that $k_i=\omega/v_i$ is the wave number in $D_i$. Similarly, the wave speed and wave number in the background medium are denoted by $v$ and $k$. The crucial asymptotic parameters are the contrast parameters $\delta_1,\dots,\delta_N$. For example, in the case of an acoustic system, $\delta_i$ is the ratio of the densities inside and outside the resonator, respectively. In the current formulation, the material parameters may take any values (for example, complex parameters correspond to non-Hermitian systems with energy gain and loss). In the setting of \Cref{sec:capacitance}, we take all parameters to be positive and equal.

Subwavelength resonance will occur in the high-contrast limit
$$\delta_i \to 0.$$
We define $D$ as the collection of resonators:
$$D=\bigcup_{m\in I_r} \bigcup_{i=1}^N (D_i+m),$$
and consider the Helmholtz resonance problem in $D$
\begin{equation} \label{eq:finite_scattering}
	\left\{
	\begin{array} {ll}
		\ds \Delta {u}+ k^2 {u}  = 0 & \text{in } \R^3 \setminus \overline{D}, \\
		\nm
		\ds \Delta {u}+ k_i^2 {u}  = 0 & \text{in } D_i+m, \text{ for } i=1,\dots,N, \ m\in I_r, \\
		\nm
		\ds  {u}|_{+} -{u}|_{-}  =0  & \text{on } \partial D, \\
		\nm
		\ds  \delta_i \frac{\partial {u}}{\partial \nu} \bigg|_{+} - \frac{\partial {u}}{\partial \nu} \bigg|_{-} =0 & \text{on } \partial D_i+m \text{ for } i=1,\dots,N, \ m\in I_r, \\
		\nm
		\multicolumn{2}{l}{\ds u(x) \ \text{satisfies the Sommerfeld radiation condition},}
	\end{array}
	\right.
\end{equation}
where the Sommerfeld radiation condition is given by
\begin{equation} \label{eq:SRC}
	\lim_{|x|\to\infty} |x|\left(\ddp{}{|x|}-\iu k\right)u=0, \quad \text{uniformly in all directions } x/|x|,
\end{equation}
and guarantees that energy is radiated outwards by the scattered solution. 

As mentioned, we take the limit of small contrast parameters while the wave speeds are all of order one. In other words, we take $\delta >0$ such that
\begin{equation}
	\delta_i=O(\delta) \quad\text{and}\quad v,v_i=O(1) \quad \text{as} \quad \delta\to0, \text{ for } i=1,\dots,N.
\end{equation}
Within this setting, we are interested in solutions $\omega$ to the resonance problem \eqref{eq:finite_scattering} that are \emph{subwavelength} in the sense that
\begin{equation}
	\omega\to0 \quad \text{as}\quad \delta\to0.
\end{equation}

To be able to characterize the subwavelength resonant modes of this system, we must define the \emph{generalized} capacitance coefficients. Recall the capacitance coefficients $(C^{mn}_\frm)_{ij}$ from \eqref{eq:Cfinite}. Then, we define the corresponding generalized capacitance coefficient as
\begin{equation} \label{eq:gcm}
	(\C_\frm^{mn})_{ij}=\frac{\delta_i v_i^2}{|D_i^m|} (C^{mn}_\frm)_{ij},
\end{equation}
where $|D_i^m|$ is the volume of the bounded subset $D_i^m$. Then, the eigenvalues of $\C_\frm$ determine the subwavelength resonant frequencies of the system, as described by the following theorem \cite{ammari2021functional}.

\begin{thm}
	Consider a system of $N|I_r|$ subwavelength resonators in $\mathbb{R}^3$. For sufficiently small $\delta>0$, there exist $N|I_r|$ subwavelength resonant frequencies $\omega_1(\delta),\dots,\omega_{N|I_r|}(\delta)$ with non-negative real parts. Further, the subwavelength resonant frequencies are given by
	$$ \omega_n = \sqrt{\lambda_n}+O(\delta) \quad\text{as}\quad \delta\to0,$$
	where $\{\lambda_n: n=1,\dots,N|I_r|\}$ are the eigenvalues of the generalized capacitance matrix $\mathcal{C}_\frm$, which satisfy $\lambda_n=O(\delta)$ as $\delta\to0$.
\end{thm}
A similar result exists for an infinite periodic structure, in terms of the eigenvalues of the {generalized} quasi-periodic capacitance matrix, as defined in \eqref{eq:Calpha}; see \cite{ammari2021functional} for details.

\section*{Acknowledgements}

The work of HA was supported by Swiss National Science Foundation grant number 200021--200307. The work of BD was supported by a fellowship funded by the Engineering and Physical Sciences Research Council under grant number EP/X027422/1. 

\bibliographystyle{abbrv}
\bibliography{essential}{}
\end{document}